\documentclass[runningheads]{llncs}

\usepackage{amsmath}
\usepackage{amssymb}
\usepackage{enumitem}

\DeclareMathAlphabet\rsfscr{U}{rsfso}{m}{n}

\def \O     	{\mathcal{O}}
\def \Z 	{\mathbb{Z}}
\def \PRIMES 	{\mathbb{P}}

\def \PoSW      {\mathsf{PoSW}}
\def \CP        {\mathtt{CP}}
\def \H    	{\mathsf{H}_x}
\def \RO	{\mathsf{H}}
\def \adv       {\mathcal{A}}
\def \solve	{\mathsf{Solve^H}}
\def \open	{\mathsf{Open^H}}
\def \gen	{\mathsf{Gen}}

\def \vdf	{\mathsf{VDF}}
\def \verify	{\mathsf{Verify}}
\def \verifyH	{\mathsf{Verify^H}}
\def \parent	{\mathsf{parent}}

\def \negl	{\mathtt{negl}}
\def \poly	{\mathtt{poly}}
\def \st	{\mathtt{state}}
\def \mod       {\; \mathbf{mod} \;}
\def \multgroup#1{(\mathbb{Z}/#1\mathbb{Z})^\times}
\def \vrf       {\mathcal{V}}
\def \prv       {\mathcal{P}}

\def \rfg       {\mathcal{G}^\frac{1}{\ell}}

\begin{document}
\title{Single-Query Verifiable Proof-of-Sequential-Work}

\titlerunning{Single-Query Verifiable $\PoSW$}

%
\author{Souvik Sur}
%
%
\institute{Department of Computer Science and Engineering,\\ 
Indian Institute of Technology Kharagpur,\\ Kharagpur, West Bengal, India \\
\email{souviksur@iitkgp.ac.in}}

\maketitle              

\begin{abstract}\label{abstract}
We propose a proof-of-sequential-work ($\PoSW$) that can be verified with only a single query 
to the random oracle for each random challenge.
Proofs-of-sequential-work are protocols that facilitate 
a verifier to efficiently verify if a prover has executed a specified number 
of computations sequentially.
Denoting this number of sequential computations with $N$, 
the prover with $\poly(N)$ parallelism must take $\Omega(N)$-sequential time 
while the verifier verifies the computation in $\O(\log N)$-sequential time using upto $\O(\log N)$ parallelism.
We propose a $\PoSW$ that allows any verifier, even the one 
with no parallelism, to verify using just a single sequential computation on a single challenge.


All the existing $\PoSW$s~\cite{Mahmoody2013Sequential,Cohen2018Simple,Abusalah2019Reversible,Dottling2019Incremental} 
mandate a prover to compute a sequence of responses 
from a random oracle against $N$-rounds of queries.
Then the prover commits this sequence using a commitment scheme 
(e.g., Merkle root (like) commitment) predefined in the $\PoSW$s.
Now the verifier asks the prover to provide a set of proofs against 
$t$ randomly chosen checkpoints, called challenges, in the computed sequence.
The verifier finds out the commitment from each of these proofs spending $\O(\log N)$ rounds of queries to the oracle. 
It can be reduced to a single round of queries 
only if the verifier owns $\O(\log N)$ parallelism~\cite{Dottling2019Incremental}.

The verifier in our $\PoSW$ demands \emph{no} parallelism but uses a \emph{single} query to the random oracle in order 
to verify each of the $t$ challenges.
The key observation is that the commitment schemes themselves in the prior works demand $\O(\log N)$ oracle queries to verify. 
So our $\PoSW$ asks the prover to undergo an additional efficient binary operation $\otimes$ on 
the responses from the random oracle against $N$-rounds of queries. 
The cumulative result of $\otimes$, represented as a map $f$, on all such responses serves the purpose of the commitment. 
The verifier verifies this cumulative result with a single query to the oracle exploiting some special properties of $f$.
Thus the prover still needs $\Omega(N)$-rounds of queries to compute the commitment but any (non-parallel) verifier needs 
only a single query to the random oracle to verify. We instantiate $\otimes$ along with $f$ practically.

We stress that the sequentiality of this proposed $\PoSW$ does not depend on the choice of the operation $\otimes$ or 
the map $f$ but proven under the random oracle model. However, its soundness demand some specific properties, 
which are specified in the end, for $\otimes$ and $f$.
\end{abstract}


\section{Introduction}\label{introduction}
The notion of proofs of sequential work ($\PoSW$) was introduced by 
Mahmoody et al. in~\cite{Mahmoody2013Sequential}. A $\PoSW$ is a protocol 
that enables a verifier $\vrf$ to efficiently check if a prover $\prv$ 
has gone through $N$ sequential steps after receiving some statement $x$.
Upon receiving $x$, $\prv$ computes some sequence spending at least 
$N$ sequential steps (i.e. time) and commits the sequence to some commitment $\phi$ using a 
specified commitment scheme. Observing the $\phi$, $\mathcal{V}$ challenges 
$\mathcal{P}$ to supply $t$ number of proofs $\pi_i$ of $\mathcal{V}$'s choice. $\mathcal{V}$ verifies 
the integrity of $\phi$ through each of the proofs $\pi_i$. If all of them are correct then $\mathcal{V}$
accepts that $\mathcal{P}$ has spent $N$ sequential steps on the input $x$; rejects otherwise. 
In order to verify efficiently, 
$\mathcal{V}$ keeps $t$ as small as possible but sufficient to catch an adversary $\adv$ 
intending to skip a fraction (say $\alpha$) of $N$ with non-negligible probability.

This fraction $\alpha$ and the sufficiency of the commitment largeness of $t$ ties an important property with 
every $\PoSW$. It is known as the soundness of a $\PoSW$. Soundness demands a $\PoSW$ 
to guarantee that no adversary $\adv$ making only $(1-\alpha)$ fraction of 
$N$ sequential steps would be accepted with the probability $>(1-\alpha)^t$. 
The soundness of all the $\PoSW$s
\cite{Mahmoody2013Sequential,Cohen2018Simple,Abusalah2019Reversible,Dottling2019Incremental} 
are proven on the assumption that given a string $x\in\{0,1\}^*$ and a random oracle $\mathsf{H}$, no adversary $\adv^\mathsf{H}$ making only $(1-\alpha)N$ queries to $\mathsf{H}$, can compute a sequence of responses 
$\mathsf{H}^i(x)= \mathsf{H}(\mathsf{H}^{(i-1)}(x))$ for $i=1,2,\ldots,N$ with the probability 
$>(1-\alpha)^t$. 

The verification, in all these $\PoSW$s, need at least $\log N$ queries to the random oracle for 
the time parameter $N$. D\"{o}ttling et al. reduces it to a single round of query only if the verifier has 
$\O(\log N)$ parallelism~\cite{Dottling2019Incremental}. Our $\PoSW$ reduces it to a single query to the random oracle even if 
the verifier has no parallelism.


\subsection{Organization of the Paper}
Section~\ref{literature} discusses all the existing $\PoSW$s. 
In Section~\ref{preliminaries}, we describe a succinct review of
the technicalities $\PoSW$ and random oracle.
Section~\ref{practical} demonstrates the design of our single-query verifiable $\PoSW$.
In Section~\ref{security} we analyze the security of the $\PoSW$.
We compare the efficiency of the proposed $\PoSW$ and the existing ones in Section.~\ref{efficiency}. 
We generalize the map $f$ in Section.~\ref{design} in order to explore the other possibilities 
to have such a $\PoSW$.
Finally, Section~\ref{conclusion} concludes the paper.

\section{Related Work}\label{literature}
With its introduction, the first $\PoSW$ by Mahmoody et al. asks a prover to compute the labels 
against all the vertices of a directed depth robust graph of $N$ nodes using a random oracle~\cite{Mahmoody2013Sequential}.
The label against a vertex requires to be computed recursively from the labels of all of its parents. 
Now the prover sends a Merkle root commitment of all the labels to the verifier.
The verifier challenges the prover on some of these labels. 
Given such a challenge node, the verifier needs to provide the labels of the challenge node, 
its parents, and the siblings of the nodes that lie over the path from the challenge node to the Merkle root, as a proof $\pi$.
The verifier finds the label of the node using its parents' labels. Then (s)he reconstructs the Merkle root using 
the labels of the siblings in $\O(\log N)$-time. 

The soundness of this $\PoSW$ is based on the property that
an $(\alpha,(1-\alpha))$-depth-robust graph always has a path of length of $(1-\alpha) N$ even after removing 
$\alpha N$ of vertices where $\alpha < 1$. 
So a prover has to evaluate the labels along a path of length $(1-\alpha) N$ spending $(1-\alpha) N$ sequential time. 

Cohen and Pietrzak propose another $\PoSW$ using a directed binary tree with some additional edges~\cite{Cohen2018Simple}. 
Each of these additional edges ends at each of the leaves of the starting from the left 
sibling of the nodes on the paths from the leaves to the root.
Essentially the idea is similar to the one in~\cite{Mahmoody2013Sequential}, however, they use the Merkle tree
not only for verification but also to guarantee the $N$ sequential computations. Moreover,
the labels in the graph can be computed in topological order with the help 
of only $\mathcal{O}(\log N)$ labels. Thus a prover spends $N$-sequential time to label the tree.
The effort for verification is $\O(\log N)$-time via Merkle root verification as mentioned above.
We discuss this work in more detail in Sect.~\ref{CP} as it is at the heart 
of our $\PoSW$. 

In the recent past, Abusalah et al. designed a reversible $\PoSW$ in~\cite{Abusalah2019Reversible} 
with the help of skip list graph where for each edge $(i,j)$ there exists a $k>0$ such that $j-i=2^k$ and $2^k \mid i$. 
The protocol asks the prover to label a skip list of $N$ nodes using a random permutation oracle.
The verifier selects some challenge nodes uniformly at random and checks if 
the paths containing the challenge nodes are consistent.
It has been shown that random oracle and random permutation oracle are indistinguishable~\cite{Impagliazzo1989Limits}.
So the prover needs $N$-sequential time to label the skip list. The verifier exploits the typical property of a skip list 
to verify in $\O(\log N)$-time.


At the same time, incremental proof-of-sequential-work was introduced by D\"{o}ttling et al.~\cite{Dottling2019Incremental}.
This additional feature allows a prover to continue the computation from some earlier checkpoints. 
Their construction is also based on the $\PoSW$ in~\cite{Cohen2018Simple}. They 
make it incremental by choosing the challenge leaves dynamically while labeling the graph. So the rest 
of the graph can be pruned gradually in the run-time. 
In order to determine the challenge leaves under a node it randomly chooses the set of leaves 
from both of its subtrees using another random oracle on the label of that node. 
So only the challenge paths are required to be stored for verification. 
The effort verification is $\O(\log N)$ as the challenge paths are of length $\O(\log N)$ at most.

The verification in all of the above $\PoSW$s can be parallelized upto the availability
of $\O(\log N)$-parallelism.


\subsection{Overview of Our $\PoSW$}
Like~\cite{Dottling2019Incremental}, our $\PoSW$ is based on 
the $\PoSW$ by Cohen and Pietrzak~\cite{Cohen2018Simple}. 
We call the $\PoSW$ in~\cite{Cohen2018Simple} as the $\CP$ construction and modify it as follows.

First we need an \emph{efficient} binary operation $\otimes$ and a binary map $f$ defined over 
the set $\{0,1\}^\lambda$ such that,
\begin{enumerate}
 \item for all $g,a$ and $b$, $f(f(g,a),b)=f(g,(a \otimes b))$.
 \item given the result $f(g,a)$ and $a$, it is computationally hard (w.r.t. $\lambda$) to find $g$. 
 \item the quantity $f(s_0,((s_1 \otimes \ldots \otimes s_N)\otimes s_{i}^{-1}))$ 
 is defined if and only if $s_i \in \{ s_1,s_2,\ldots,s_N \}$.
\end{enumerate}
We instantiate $\otimes$ with the multiplication over integers, $f$ as modulo exponentiation and $s_i$ as 
the $i$-th prime in our $\PoSW$ in Sect.~\ref{practical}. However, we prefer to discuss the fundamental idea using 
$\otimes$ and $f$ showing that the elegance of our $\PoSW$ is independent of this instantiation.

The evaluation phase in our $\PoSW$ works exactly as $\CP$ except that it asks the prover to evaluate 
the function $f$ on the labels (labeled with random oracle) of each node of $\CP$-graph 
(discussed in Sect.~\ref{CP}). 
For the time parameter $N$ we need a $\CP$-graph $G_n=(N,E)$ with $N=2^{n+1}-1$ nodes w.l.o.g. for some $n \in \Z$.
The security parameter $\lambda$ determines the random oracle $\mathsf{H} : \{0,1\}^*\rightarrow \{0,1\}^\lambda$.

Given the input $x\in\mathcal{X}$ the prover samples a random oracle 
$\H(\cdot)\stackrel{def}{=}\mathsf{H}(x \| \cdot)$.
Now (s)he computes a $\lambda$-bit label $s_0=\H(0^{n})$.  
Then (s)he labels the entire graph $G_n$  
as $s_{i}=\H(i\| s_{k_1}\| \ldots \|s_{k_j})$ 
where the $\{k_j\}$s are the parents of the node $i$.

Along with with the labeling, our $\PoSW$ also asks the prover to compute the product 
$\rho=(s_1\otimes s_2\otimes\cdots\otimes s_{N})$ and the commitment 
$\phi=f(s_0,\rho)=f(\ldots f(f(s_0,s_1),s_2)\ldots s_{N})$. 
The prover is allowed to store any fraction (even fully) of all the labels $\{s_1, s_2,\ldots, s_{N}\}$. 

During verification, the verifier chooses $t$ challenge leaves 
$\{\gamma_1,\gamma_2,\ldots,\gamma_t\}$ of the graph $G_n$, uniformly at random.
For each leaf $\gamma_i$, 
the verifier asks the prover for the proof $\pi_i=(\sigma_i,\tau_i)$ where $\sigma_i=\{s_{k_1}, s_{k_2}, \ldots, s_{k_j}\}$ 
such that $s_{\gamma_i}=\H(\gamma_i\| s_{k_1}\| \ldots \|s_{k_j})$ and $\tau_i=f(s_0,(\rho\otimes s_{\gamma_i}^{-1}))$.
The verifier checks if $s_{\gamma_i}\stackrel{?}{=}\H(\gamma_i\| s_{k_1}\| \ldots \|s_{k_j})$ and if 
$\phi\stackrel{?}{=}f(\tau_i,s_{\gamma_i}^{-1})=f(s_0,(\rho\otimes s_{\gamma_i}^{-1}))$.
The quantity $f(s_0,(\rho\otimes s_{\gamma_i}^{-1}))$ is defined if and only if 
$s_{\gamma_i}$ is one of the labels of the $\CP$ graph. The integrity of the label $s_{\gamma_i}$ is confirmed 
by the \emph{single} query $s_{\gamma_i}\stackrel{?}{=}\H(\gamma_i\| s_{k_1}\| \ldots \|s_{k_j})$.
If both the checks are true for all the challenge leaves $\{\gamma_1,\gamma_2,\ldots,\gamma_t\}$ 
then verifier accepts it, rejects otherwise. 

Sect.~\ref{contributions} presents a detailed comparison 
among our design with the existing ones.

\section{Preliminaries}\label{preliminaries}

We fix the notations first.
\subsection{Notations} 
We take $\mathcal{P}$ and $\mathcal{V}$ as the prover and the verifier respectively.
We denote the security parameter with $\lambda\in\mathbb{Z}^+$ 
and the sequential time parameter $N\in\mathbb{Z}^+$. Here $\poly(\lambda)$ is some function $\lambda^{\O(1)}$, and
$\negl(\lambda$) represents some function $\lambda^{-\omega(1)}$. 

For some $x,z\in\{0,1\}^*$, $x\|z$ implies the concatenation of elements $x$ and $z$.
When $x\in\{0,1\}^*$ is a string then $|x|$ denotes its bitlength.
The alphabets $\mathbb{T},\mathbb{S},\mathbb{R}$ and $\mathbb{U}$ represents sets defined in the context. 
We denote $|\mathbb{S}|$ as the cardinality of $\mathbb{S}$. 

If any algorithm $\adv$ outputs $y$ on an input $x$, 
we write $y\leftarrow\adv(x)$. By $x\xleftarrow{\$}\mathcal{X}$,
we mean that $x$ is sampled uniformly at random from $\mathcal{X}$.
We consider $\adv$ as efficient if it runs in 
probabilistic polynomial time (PPT) in $\lambda$.
We assume (or believe) a problem to be hard if it is yet to have an efficient 
algorithm for that problem.
We denote $\mathsf{H}:\{0,1\}^*\rightarrow\{0,1\}^w$ as a random oracle.  
If an algorithm $\adv$ queries the random oracle $\mathsf{H}$ it is denoted 
as $\adv^\mathsf{H}$.

\subsection{Random Oracle}
\begin{definition}{\bf (Random Oracle $\mathsf{H}$).}\label{RO}
A random oracle \textsf{H}:$\{0,1\}^*\rightarrow\{0,1\}^w$ is a map that always 
maps any element from its domain to a fixed element chosen uniform at random from 
its range.
\end{definition}
%


\begin{definition}{\bf ($\RO$-sequence).}\label{sequence}
 An $\RO$-sequence of length $\mu$ is a sequence
$x_0, x_1,\ldots,x_\mu\in\{0,1\}^*$ where for each $i$, $1 \le i < \mu, 
{\normalfont \textsf{H}}(x_i)$ is contained in $x_{i+1}$
as continuous substring, i.e., $x_{i+1} = a||{\normalfont \textsf{H}}(x_i)||b$ for some $a,b\in\{0,1\}^*$.
\end{definition}

We mention the following theorem from~\cite{Cohen2018Simple} for the sake of completeness.
\begin{lemma}{\bf ($\mathsf{H}$ is Sequential).}\label{thm:Sequential}
 With at most $(N-1)$ rounds of queries to {\normalfont\textsf{H}}, 
 where in each round one can make arbitrary many parallel queries. 
 If $\mathsf{H}$ is queried with at most $\mu$ queries of total length $Q$ bits, then
 the probability that {\normalfont\textsf{H}} outputs an $\mathsf{H}$-sequence 
 $x_0,\ldots,x_N\in\{0,1\}^*$ is at most $$\mu .\frac{Q+\sum\limits_{i=0}^{N}|x_i|}{2^{w}}$$.
\end{lemma}
\begin{proof}
 There are two ways to figure out an $\mathsf{H}$-sequence 
 $x_0,\ldots,x_\mu\in\{0,1\}^*$ with only $(N-1)$ sequential queries.
 \begin{description}
  \item [{\normalfont Random Guess:}] As $\mathsf{H}$ is uniform, for some $a,b\in\{0,1\}^*$ and some 
  $i$, if the query$\mathsf{H}(x_i)$ has not been made then, 
  $$\Pr[x_{i+1} = a||\mathsf{H}(x_i)||b \;]\le \mu .\frac{\sum\limits_{i=0}^{N}|x_i|}{2^{w}}$$
  \item [{\normalfont Collision:}] If $x_i$'s were not computed sequentially then, 
for some $1 \le i \le j \le N - 1$, a query $a_i$ is made in round $i$ and
query $a_j$ in round $j$ where $\mathsf{H}(a_j)$ is a sub-element of $a_i$.
As {\normalfont \textsf{H}} is uniform, the probability of this event $\le \mu .\frac{Q}{2^w}$.
 \end{description} 
  
\end{proof}

\subsection{Proof of Sequential Work}\label{protocol}
Mahmoody et al. are the first to formalize the idea of $\PoSW$s in~\cite{Mahmoody2013Sequential}.
All the existing $\PoSW$s are defined in the random oracle model as 
they inherit their sequentiality from that of the random oracle.
However, it is not necessary as there exist other sources of sequentiality e.g.,
time-lock or RSW puzzle~\cite{Rivest1996Time}. Therefore, we define $\PoSW$ in general,
irrespective of the random oracle model. Later, we translate the
same definition in the random oracle model in Sect.~\ref{PoSWRO}. 

\begin{definition}\normalfont{ \bf (Proof of Sequential Work $\PoSW$).}
Assuming $\mathcal{X},\mathcal{Y}\subseteq\{0,1\}^*$, a $\PoSW$ is a quadruple of algorithms 
$\mathsf{Gen, Solve, Open, Verify}$ 
that implements a mapping $\mathcal{X}\rightarrow\mathcal{Y}$ as follows,
\begin{description}
\item $\mathsf{Gen}(1^\lambda,N) \rightarrow \mathbf{pp}$
 is an algorithm that takes as input a security parameter $\lambda$ and time parameter $N$
 and produces the public parameters $\mathbf{pp}$. These parameters $\mathbf{pp}$ are implicit 
 in each of the remaining algorithms.
 
 \item $\mathsf{Solve}(\mathbf{pp},x)\rightarrow(\phi,\phi_\mathcal{P})$ 
 takes an input $x\in\mathcal{X}$ usually called statement, and produces a commitment $\phi\in\mathcal{Y}$. 
 The triple $(x,N,\phi)$, often called commitment, is publicly announced by $\mathcal{P}$. 
 Additionally $\mathcal{P}$ may produce some extra information $\phi_\mathcal{P}\in\{0,1\}^*$ and stores it locally 
 in order to use in the $\mathsf{Open}$ algorithm. 
 The running time of $\mathsf{Solve}$ must be at least $N$.
 
 \begin{description}
 \item[Challenge Vector]
 Observing an announced triple $(x,N,\phi)$, $\mathcal{V}$ 
 samples a challenge vector $\gamma=\{\gamma_1,\gamma_2,\ldots,\gamma_t\}\in_R\mathbb{Z}^t_N$ uniformly at random.
 \end{description}
 \item $\mathsf{Open}(\mathbf{pp},x,N,\phi,\phi_\mathcal{P},\gamma)\rightarrow\pi$ 
 takes the challenge vector $\gamma$ and the locally stored information $\phi_\mathcal{P}$ as the inputs, 
 and sends a proof vector $\pi=\{\pi_1,\pi_2,\ldots,\pi_t\}\in\{0,1\}^*$ to $\mathcal{V}$. Essentially, $\mathcal{P}$
 runs $\mathsf{Open}$ to generate each $\pi_i$ corresponds to each $\gamma_i$.

 \item $\mathsf{Verify}(\mathbf{pp}, \pi,x,N,\gamma,\phi)\rightarrow \{0, 1\}$ is an
 algorithm that takes a triple $(x,N,\phi)$, a challenge vector $\gamma$, a proof vector $\pi$
 and either accepts ($1$) or rejects ($0$). We say the commitment triple is a valid one if and only if $\mathsf{Verify}$
 accepts it. The algorithm must be ``exponentially" faster than 
 $\mathsf{Solve}$, in particular, must run in $\poly(\lambda,\log{N})$ time.
\end{description}
\end{definition}

%
%
%

Before we proceed to the security of a $\PoSW$ we precisely model parallel adversaries that 
suit the context. 
\begin{definition}\normalfont{(\bf Parallel Adversary)}\label{adversaries} 
A parallel adversary $\adv=(\adv_0,\adv_1)$ is a pair of non-uniform 
randomized algorithms $\adv_0$ with total running time $\poly(\lambda,N)$, 
and $\adv_1$ which runs in parallel time $\delta<N-o(N)$ on at 
most $\poly(\lambda,N)$ number of processors.
\end{definition}
Here, $\adv_0$ is a preprocessing algorithm that precomputes some
$\st$ based only on the public parameters, and $\adv_1$ exploits
this additional knowledge to solve in parallel running time $\delta$ on 
$\poly(\lambda,N)$ processors.

The three necessary properties of a $\PoSW$ are now introduced.

\begin{definition}\normalfont{(\bf Correctness)}\label{def: Correctness} 
A $\PoSW$ is correct, if for all $n,N,\mathbf{pp}$, 
and $x\in\mathcal{X}$, we have
\[
\Pr\left[
\begin{array}{l}
\mathsf{Verify}(\mathbf{pp},\pi,x,\gamma,\phi)=1
\end{array}
\Biggm| \begin{array}{l}
\mathbf{pp}\leftarrow\mathsf{Gen}(1^\lambda,N)\\
x\xleftarrow{\$}\mathcal{X}\\
(\phi,\phi_\mathcal{P})\leftarrow\textsf{Solve}(\mathbf{pp},x)\\
\pi\leftarrow\textsf{Open}(\mathbf{pp},x,\phi_\mathcal{P},\gamma)
\end{array}
\right]
=1.
\]
$\mathcal{V}$ always accept a triple $(x,N,\phi)$ generated by $N$ sequential queries to \textsf{H}.
\end{definition}

\begin{definition}\normalfont{\bf(Soundness)}\label{def: Soundness} 
A $\PoSW$ is sound if for all non-uniform parallel algorithms $\adv$ 
(Def.~\ref{adversaries}) that run in $(1-\alpha)N$ 
time, for some $0 < \alpha <1$, we have

\[
\Pr\left[
\begin{array}{l}
\phi\ne\mathsf{Solve}(\mathbf{pp},x)\\
\mathsf{Verify}(\mathbf{pp},\pi,x,\gamma,\phi)=1
\end{array}
\Biggm| \begin{array}{l}
\mathbf{pp}\leftarrow\mathsf{Gen}(1^\lambda,N)\\
\st\leftarrow\adv_0(1^\lambda,N,\mathbf{pp})\\
x\xleftarrow{\$}\mathcal{X}\\
(\phi,\pi)\leftarrow\adv_1(\st,x)\\
\end{array}
\right]\\
\le (1-\alpha)^t.
\]
Using $t$ number of random challenges $\gamma$, the verifier $\mathcal{V}$ should catch
all non-uniform parallel adversaries $\adv$ with ``non-negligible" probability.  
\end{definition}

\begin{definition}\normalfont{\bf (Sequentiality)}\label{def: Sequentiality}
A $\PoSW$ is $\delta$-sequential if for all parallel algorithms $\adv=(\adv_0,\adv_1)$ 
(Def.~\ref{adversaries}) that finds a $\phi$ in parallel time $\delta(N)< N$, it holds that
\[
\Pr\left[
\begin{array}{l}
\phi=\mathsf{Solve}(\mathbf{pp},x)
\end{array}
\Biggm| \begin{array}{l}
\mathbf{pp}\leftarrow\mathsf{Gen}(1^\lambda,N)\\
\st\leftarrow\adv_0(1^\lambda,N,\mathbf{pp})\\
x\xleftarrow{\$}\mathcal{X}\\
\phi\leftarrow\adv_1(\st,x)
\end{array}
\right]
\le \negl(\lambda).
\]
\end{definition}

\begin{description}

\item[Non-interactive $\PoSW$s] The $\mathsf{Open}$ phase is required only in the interactive version of a $\PoSW$. 
  It can be made non-interactive using another hash function 
  $\H':\{0,1\}^*\rightarrow\mathbb{Z}_N$ as per the Fiat-Shamir heuristic. In that case, $\mathcal{P}$ 
  will compute $$\gamma=\{\H'(\phi\|1), \H'(\phi\|2),\ldots,\H'(\phi\|t)\}.$$

 \item[Subexponentiality of Time $N$] \label{subexp}
 An adversary $\adv$ running on $\poly(\lambda,2^{\O(\lambda)})$ processors
will always be able to efficiently find a valid commitment $(x,N,\phi)$ for any $N\in 2^{\O(\lambda)}$.
The trick is to brute-force the proof space using $\mathsf{Verify}$ which is efficient.
Given a statement $x$ and target time $N$, $\adv$ does not need to run $\mathsf{Solve}$ 
rather (s)he will choose a $\phi\in_R \mathcal{Y}$ uniformly at random.
Now $\mathcal{V}$ will sample a challenge vector $\gamma$ for which $\adv$ needs to find a proof vector $\pi$
such that, $\mathsf{Verify}(\mathbf{pp},\pi,x,\gamma,\phi)=1$.
For each $\gamma_i$, $\adv$ will run $2^{\O(\lambda)}$ instances of $\mathsf{Verify}$ each with a different $\pi'$ 
on each of its processors and identify the correct $\pi'$ with $\mathsf{Verify}(\mathbf{pp},\pi',x,\gamma,\phi)=1$.
So $\PoSW$s restrict $N \in 2^{o(\lambda)}$ enforcing the complexity of this brute-force approach 
to be $2^\lambda/2^{o(\lambda)}=2^{\Omega(\lambda)}$.

\item[$\PoSW$ in the Random Oracle Model]\label{PoSWRO}
$\PoSW$s are not necessarily defined in the random oracle model, however traditionally all the existing $\PoSW$s have been 
so~\cite{Mahmoody2013Sequential,Cohen2018Simple,Abusalah2019Reversible,Dottling2019Incremental}.
Essentially in the random oracle model both the prover $\mathcal{P}$, the verifier $\mathcal{V}$ and also 
the parallel adversary $\adv$ are allowed to access a common random oracle $\mathsf{H}$. So we write 
$\mathsf{PoSW^H=\{Gen, Solve^H, Open^H, Verify^H}\}$ 
to emphasize that all these algorithms except $\mathsf{Gen}$ may access 
the random oracle $\mathsf{H}$. The input statement $x$ samples a random oracle $\H$ 
where $$\H(\cdot)\stackrel{def}{=}\mathsf{H}(x \| \cdot).$$
The notion of sequentiality in these $\PoSW$s comes from the fact that $\mathsf{Solve^H}$ 
requires $\H$-sequence of length $N$ to compute the commitment $(x,N,\phi)$. By lemma~\ref{thm:Sequential}, 
no parallel adversary $\adv$ running on $\poly(\lambda,N)$ processors, can produce a $(x,N,\phi')$ in 
time $<N$ such that $\Pr[\phi=\phi']>\negl(\lambda)$. 
On the contrary, $\mathsf{Verify^H}$ uses only $\poly(\lambda,\log N)$ queries to \textsf{H}. So we call such a 
$\PoSW$ as a proof-of-sequential-work in the random oracle model.

\end{description}


\subsection{The $\CP$ $\PoSW$ }\label{CP}
It has two parts.
\subsubsection{The $\CP$ Graph}
Suppose $N=2^{n+1}-1$ and $B_n=(V,E')$ is a complete binary tree of depth $n$ 
with edges pointing to the root from the leaves. So the set $V$ can be identified with $\{0,1\}^{\le n}$
binary strings of length $\le n$, identifying root with the null string $\epsilon$. 
As the edges point upward each of the internal nodes have $2$ parents, left and right. 
The index of the left and right parents of a node $v$ are $v\|0$ and $v\|1$, respectively. 
Therefore essentially,
$$E'=\{(v\|0,v)\cup (v\|1,v) \mid v \in  \{0,1\}^{< n}\}.$$
The leaves are identified with $\{0,1\}^{n}$. 
So, the node $w$ lies over the path from a node $u$ to the root if $u=w\|x$ for some $x\in \{0,1\}^{n-|w|}$.

The $\CP$ graph $G_n=(V,E' \cup E'')$ is essentially the graph $B_n$ with the additional edges,
$$E''=\{(v,u)\mid u \in \{0,1\}^n, u=w\|1\|w', v=w\|0 \}.$$

It means an edge $(v,u)\in E''$ if and only if the node $v$ is a left sibling of another node that lies 
over the path from the leaf $u$ to the root. We denote $\parent(v)=\{u \mid (u,v) \in E' \cup E''\}$.

\subsubsection{The $\CP$ Protocol}
Given the input $x\in\mathcal{X}$ the prover samples a random oracle 
$\H(\cdot)\stackrel{def}{=}\mathsf{H}(x \| \cdot)$.
Now (s)he computes a $\lambda$-bit label $s_0=\H(0^{n})$.  
Then (s)he labels the entire graph $G_n$  
as $s_{i}=\H(i\| s_{k_1}\| \ldots \|s_{k_j})$ 
where the $\parent(i)=\{k_1, \ldots, k_j\}$. 
The label of the root $\phi$ serves the purpose of the commitment.

The verification exploits two important properties of $G_n$,

\begin{enumerate}
 \item Given a leaf $v$, the labels of $\parent(v)$ are necessary and 
 sufficient to compute the label of the root $\phi$.
 \item For any $0 < \alpha <1$, there exists a path of length $(1-\alpha)N$ in the induced subgraph of $G_n$ having 
 $(1-\alpha)N$ nodes.
\end{enumerate}

The verifier chooses $t$ challenge leaves 
$\{\gamma_1,\gamma_2,\ldots,\gamma_t\}$ of the graph $G_n$, uniformly at random.
For each leaf $\gamma_i$, 
the verifier asks the prover for the proof $\pi_i=\parent(\gamma_i)$.
The verifier finds the label of the root $\phi'$ using the labels of 
the node $\parent(\gamma_i)$ using the first property of $G_n$.
If the committed and the computed labels of the root match i.e., $\phi=\phi'$ then 
the verifier accepts, rejects otherwise. 
The soundness claim comes from the second property of $G_n$. 
An adversary has to label a path of length $(1-\alpha)N$ if 
(s)he attempts to skip $\alpha N$ nodes. 
The sequentiality of this $\PoSW$ stands on the sequentiality of $\mathsf{H}$.
It takes at least $(1-\alpha)N$ sequential time to label a path of length $(1-\alpha)N$.

\section{Single Query Verifiable $\PoSW$}\label{practical}
Here we present our $\PoSW$ that verifies each challenge with only a single query to the random oracle. 
We denote $\lambda\in\mathbb{Z}$ as the security parameter
and  $N$ as the targeted sequential steps.
The four algorithms for this $\PoSW$ are, 

\subsection{The $\gen(1^\lambda,N)$ Algorithm}
The generated public parameters are
$\mathbf{pp}=( \RO,G_n,\times,f,t)$ having the following meanings. 
\begin{enumerate}

\item $\mathsf{H}:\{0,1\}^*\rightarrow\PRIMES_\lambda$ is a random oracle that maps any
arbitrary binary strings to the set of first $2^\lambda$ primes each denoted as $p_i$.


\item  $G_n$ is a $\CP$-graph having $N=2^{n+1}-1$ nodes (w.l.o.g.).

\item $\times:\mathbb{R}\times\mathbb{R}\rightarrow\mathbb{R}$ represents the
multiplication over the real numbers. We observe that,  
 \begin{enumerate}
 \item $\langle \mathbb{R} , \times \rangle $ forms a group but 
 $\langle \mathbb{Z} , \times \rangle $ forms a monoid. 
 \item It allows \emph{efficient} computation of,
 \begin{enumerate} \item the product $(a\times b)$ for all $a,b\in\mathbb{R}$.
                 \item the inverse $a^{-1}$ for all $a\in\mathbb{R}$.
                \end{enumerate}
 \item For any subset $\mathcal{S}_k=\{p_0, p_1, \ldots, p_k\}\subseteq\PRIMES^k_\lambda$ the product 
 $((p_0\times \ldots\times p_k)\times p_i^{-1})\in\mathbb{Z}$ if and only if $p_i\in\mathcal{S}_k$.
 \end{enumerate}
\item We define $f:\Z^+\times\Z^+\rightarrow\multgroup{\Delta}$ 
 as $f(g,a)=g^{a}\mod \Delta$ 
 where $\Delta=pq$ is a product of two safe primes that needs $\Omega(N)$-time to be factored.
 The choices for $\Delta$ has been reported in Table~\ref{tab : RSA}.
 We stress that the integer factorization of $\Delta$ is known 
 to neither $\mathcal{P}$ nor $\mathcal{V}$. Secrecy of this factorization is important 
 as $\mathcal{P}$ may violate the soundness of this $\PoSW$ with this knowledge 
 (See Lemma~\ref{factorization}). However, this $\PoSW$ is a public coin because $\mathcal{V}$ uses no 
 secret information. 
 
 The map $f$ requires these three properties. 
\begin{enumerate}
\item For all $a$ and $b$, ${((g)^a)}^b\;\mathbf{mod}\; \Delta=g^{(a\times b)}\;\mathbf{mod}\; \Delta$.
\item Given the result $f(g,a)=g^{a}\mod \Delta$ and $a$, finding $g$ is known 
 as the Root Finding Problem (Def.~\ref{lthroot}). This problem is believed to be hard in 
 the multiplicative group $\multgroup{\Delta}$ and thus
 determines the size of $\Delta$ as a function of $N$ (See Table~\ref{tab : RSA}).
\item By the third property of the operation $\times$, 
 the quantity $f(p_0,((p_1\times p_2 \times \ldots\times p_k)\times p_i^{-1})$ is defined 
 if and only if $p_i\in\{p_1, p_2, \ldots, p_k\}$,
 as $f$ operates on the domain $\Z$ only.
 \end{enumerate}

 \item $t$ is the total number of random challenges.
\end{enumerate}
 
 Here we mention two additional properties, specific to this choice for $f(g,a)=g^{a}\mod \Delta$,
 that are not necessary for our $\PoSW$ but come for free.
 \begin{description}
 \item [Discrete Logarithm Problem] Given the result $f(g,a)=g^{a}\;\mathbf{mod}\; \Delta$ and $g$ is another famous problem,
 namely the Discrete Logarithm Problem, which is believed to be hard under the assumption that 
 the integer factorization of $\Delta$ is unknown.
 \item [Time-lock Puzzle] Modulo exponentiation $g^{a}\;\mathbf{mod}\; \Delta$ is believed to be a sequential operation that 
 takes $\Omega(\log_2 a)$-time when the integer factorization of $\Delta$ is unknown. It is renowned as the time-lock or 
 RSW assumption~\cite{Rivest1996Time}.
 As mentioned already, our $\PoSW$ stands sequential even if $f$ is an $\O(1)$-time 
 operation. So we do not count the sequentiality of $f$ in the required sequential
effort. We elaborate on this issue 
 in Sect.~\ref{source}.
 \end{description}

In what follows, right hand sides of all the $\leftarrow$
represent the actual computations and 
that of all the 
$=$ explain their mathematical equivalents. 
\subsection{The $\solve(\mathbf{pp},x)$ Algorithm}
$\prv$ does the following,
\begin{enumerate}
 \item sample a random oracle $\H(\cdot)\stackrel{def}{=}\mathsf{H}(x \| \cdot)$ using $x$.
 \item compute $p_0\leftarrow \H(0^n)$.
 \item initialize $\rho\leftarrow 1$ and $\phi\leftarrow f(p_0,1)=p_0 \mod \Delta$.
\item  repeat for $1 \le i \le N$:
\begin{enumerate}
 \item $p_{i}\leftarrow\H(i\| p_{k_1}\| \ldots \|p_{k_j})$ such that $\parent(i)=\{k_1, \ldots, k_j\}$.
 \item compute $\rho\leftarrow \rho\times p_{i}$.
 \item compute $\phi \leftarrow f(\phi,p_i)=(\phi)^{p_{i}}\mod \Delta$.
\end{enumerate}
\end{enumerate}

Here the $\rho=(p_1 \times p_2 \times \ldots \times p_N)$ and the commitment is 
$\phi=(p_0)^{p_1 \times p_2 \times \ldots \times p_N}\mod \Delta =(p_0)^\rho\mod \Delta$. 


\subsection{The $\open(\mathbf{pp},x,N,\phi,\phi_\mathcal{P},\gamma)$ Algorithm}
In this phase $\vrf$ samples a set of challenge leaves 
$\gamma=\{\gamma_1, \gamma_2,\ldots,\gamma_{t}\}$ where each $\gamma_i\xleftarrow{\$}\Z_N$. 
For each $\gamma_i$, $\mathcal{P}$ needs to construct the proof 
$\pi=\{\sigma_i,\tau_i\}_{1 \le i \le t}$ by the following,

\begin{enumerate}
\item repeat for $1 \le i \le t$:
\begin{enumerate}
 \item assign $\sigma_i \leftarrow  \parent(\gamma_i)$.
 \item $p_{\gamma_i}\leftarrow \H(i\| p_{k_1}\| \ldots \|p_{k_j})$ s.t $\parent(\gamma_i)=\sigma_i=\{k_1, \ldots, k_j\}$.
 \item compute $\tau_i  \leftarrow f(p_0,(\rho \times p_{\gamma_i}^{-1}))=p_0^{ (\rho \times p_{\gamma_i}^{-1})}\mod \Delta$.
 \item assign $\pi_i \leftarrow \{\sigma_i,\tau_i\}$.
\end{enumerate}
\end{enumerate}

$\pi$ is the proof. 
$\prv$ can not compute $\tau_i\leftarrow f(\phi, p_{\gamma_i}^{-1})$ because 
the inverse $p_{\gamma_i}^{-1}\notin\Z$. Also
without the knowledge of $\rho$ it is hard to find a $\tau_i$ 
as it requires to invert $f$ on its first argument (root finding assumption). 
As mentioned in Sect.~\ref{preliminaries}, in the non-interactive version of the protocol $\mathcal{P}$ 
will compute $\gamma=\{\H'(\phi\|1), \H'(\phi\|2),\ldots,\H'(\phi\|t)\}$.

\subsection{The $\verifyH(\mathbf{pp},x,N,\gamma,\pi,\phi)$ Algorithm}\label{alg : verify}
$\vrf$ runs these steps.
\begin{enumerate}
\item sample a random oracle $\H(\cdot)\stackrel{def}{=}\mathsf{H}(x \| \cdot)$ using $x$.
 \item compute $p_0\leftarrow  \H(0^n).$
\item repeat for $1 \le i \le t$:
\begin{enumerate}
 \item $p_{\gamma_i}\leftarrow \H(i\| p_{k_1}\| \ldots \|p_{k_j})$ s.t $\parent(\gamma_i)=\sigma_i=\{k_1, \ldots, k_j\}$.
 \item set a flag $r_i$ if and only if $\phi=f(\tau_i,p_{\gamma_i})=(\tau_i)^{p_{\gamma_i}}\mod \Delta$.
\end{enumerate}
 \item accept if $(r_1 \land r_2 \land \ldots \land r_t)\stackrel{?}{=}\top$, rejects otherwise.
  \end{enumerate}


\section{Security of This $\PoSW$}\label{security}
Three security properties of this single query verifiable $\PoSW$ are,
\subsection{Correctness} 
According to Def.~\ref{def: Correctness}, any $\PoSW$ 
should always accept a valid proof.
The following theorem establishes this correctness property of our $\PoSW$ scheme.
\begin{theorem}\label{thm: correctness}
 The proposed {\normalfont $\PoSW$} is correct.
\end{theorem}
\begin{proof}
 Since $\mathsf{H}$ always outputs a fixed element,
 $p_0$ is uniquely determined by the input $x$. Moreover in $\verify$, 
 for each $1 \le i < t$,

\begin{align*}
f(\tau_i,p_{\gamma_i})
&= (\tau_i)^{p_{\gamma_i}}\mod \Delta\\
&=p_0^{ (\rho \times p_{\gamma_i}^{-1})\times p_{\gamma_i}}\mod \Delta\\
&=p_0^\rho\mod \Delta\\
&= \phi.
\end{align*}
So, the correct labeling of $G_n$ and evaluation of $\phi$ assigns $r_i=\top$.
 Then for all $i$, $r_i=\top$ results into $r=\top$ in 
 $\verify$. So $\vrf$ has to accept it. 
 \qed
 \end{proof}
 It therefore follows that
\[
\Pr\left[
\begin{array}{l}
\verify(\mathbf{pp},x,\phi,\gamma,\pi)=1
\end{array}
\Biggm| \begin{array}{l}
\mathbf{pp}\leftarrow\gen(1^\lambda,N)\\
x\xleftarrow{\$}\mathcal{X}\\
(\phi,\phi_\mathcal{P})\leftarrow \solve(\mathbf{pp},x)\\
\pi\leftarrow\open(\mathbf{pp},x,\phi,\phi_\mathcal{P},\gamma)
\end{array}
\right]
=1.
\]

\subsection{Soundness}
The soundness warrants that an adversary $\adv$ having even $\poly(\lambda,N)$ processors cannot produce an 
invalid commitment $\phi'$ that convinces the verifier with non-negligible probability. 
We establish the soundness under the adaptive root assumption for
$\multgroup{\Delta}$~\cite{Wesolowski2019Efficient}.

\begin{definition}\normalfont{(\bf $\ell$-th Root Finding Game $\rfg$)}\label{lthroot} 
Let $\adv=(\adv_0,\adv_1)$ be a player playing the game. 
The $\ell$-th root finding game $\rfg$ goes as follows: 
\begin{enumerate}
 \item $\adv_0$ is given the output $(\RO, G_n,\times,f,t)\leftarrow \gen(1^\lambda,N)$.
 \item $\adv_0$ chooses an element $w \in \multgroup{\Delta}$ and computes some information $\st$.
 \item a prime $\ell \xleftarrow{\$}\PRIMES_\lambda$ is sampled uniform at random from the set of first $2^\lambda$ primes.
 \item observing $\ell$ and $\st$, $\adv_1$ outputs an element $v\in\multgroup{\Delta}$.
\end{enumerate}
The player $\adv$ wins the game $\rfg$ if $v^\ell=w\mod \Delta$.
\end{definition}

This problem is shown to be hard in the generic group
model~\cite{ARA2019GGM}.
In the game $\rfg$, $\adv$ obtains the $\Delta=pq$ from the description $f$ but learns nothing
about the safe primes $p$ and/or $q$. So $\adv$ does not know the order $(p-1)(q-1)$
of the group $\multgroup{\Delta}$. Thus this problem $\rfg$ is also hard, but its
relationship with any standard hard problems like integer 
factorization of $\Delta$ or RSA problem is still unknown. However, a more generic
problem of finding the $k$-th root over $\multgroup{\Delta}$ is believed to be hard.
For $k=2$, there exists a randomized reduction from the factoring of $\Delta$ to this
problem~\cite{Karl16Square}.
Observe that, this happens in the game $\rfg$ with the probability
$\Pr[\ell=2]=1/| \PRIMES_\lambda |=\negl(\lambda)$.
For arbitrary $k$, this effort can be reduced to 
$L_\Delta(\frac{1}{3},\sqrt[3]{\frac{32}{9})}$
\footnote{where $L_\Delta(\beta,\delta)=
\mathtt{exp}(\delta(1+o(1))\O(\log^{\beta}\Delta \log \log^{1-\beta}\Delta))$.} 
than 
$L_\Delta(\frac{1}{3},\sqrt[3]{\frac{64}{9})}$
for factoring, 
however, only in the presence of subexponential access to an oracle
that provides affine modular roots~\cite{Joux07Root}.
Thus the adaptive root assumption holds true for $\multgroup{\Delta}$ too.

\begin{definition}\normalfont{(\bf Adaptive Root Assumption for $\multgroup{\Delta}$)} \label{adaptive} 
There exists no efficient algorithm $\adv=(\adv_0,\adv_1)$ that wins the $\ell$-th Root Finding Game
$\rfg$ (Def.~\ref{lthroot}) with non-negligible probability in the security parameter $\lambda$.

\[
\Pr\left[
\begin{array}{l}
v^\ell=w \ne 1
\end{array}
\Biggm| \begin{array}{l}
\mathbf{pp}\leftarrow\mathsf{Gen}(1^\lambda,N)\\
(w,\st)\leftarrow\adv_0(1^\lambda,N,\mathbf{pp})\\
\ell\xleftarrow{\$}\PRIMES_\lambda\\
v\leftarrow\adv_1(w,\st)\\
\end{array}
\right]\\
\le \negl(\lambda).
\]
\end{definition}

However, we need another version of the game $\rfg$ in order to reduce the
soundness-breaking game for our $\PoSW$.

\begin{definition}\normalfont{(\bf $\ell$-th Root Finding Game with Random Oracle
$\rfg_{\mathsf{H}}$)}\label{lthrootH} 
Let $\adv=(\adv_0,\adv_1)$ be a player playing the game. 
The $\ell$-th root finding game with random oracle $\rfg_\mathsf{H}$ goes as follows: 
\begin{enumerate}
 \item $\adv_0$ is given the output $(\RO, G_n,\times,f,t)\leftarrow \gen(1^\lambda,N)$.
 \item $\adv_0$ chooses an element $w \in \multgroup{\Delta}$, $x\in \mathcal{X}$ and
$u\in \{0,1\}^*$. Also $\adv_0$ computes some information $\st$.
 \item a prime $\ell\leftarrow\H(u)$ is sampled.
 \item observing $\ell$ and $\st$, $\adv_1$ outputs an element $v\in\multgroup{\Delta}$.
\end{enumerate}
The player $\adv$ wins the game $\rfg_\mathsf{H}$ if $v^\ell=w\mod \Delta$.
\end{definition}

The computational hardness of both the games $\rfg$ and $\rfg_\mathsf{H}$ 
are equivalent as the responses of $\mathsf{H}$ are sampled uniformly at random from
$\PRIMES_\lambda$. Thus the adaptive root assumption (Def.~\ref{adaptive}) holds for the
game $\rfg_\RO$ also.

\begin{theorem}\label{thm: soundness}
 Suppose $\adv$ be an adversary who breaks the soundness of this proposed $\PoSW$
 with probability $p_{win}$. Then there exists \emph{one} of these two attackers,
 \begin{enumerate}
  \item $\sqrt{\adv}$ winning the root finding game with random oracle $\rfg_\mathsf{H}$, 
  \item $\adv_{\CP}$ breaking the soundness of the $\CP$ construction,
 \end{enumerate}
 with the same probability $p_{win}$.
\end{theorem}

\begin{proof}
We give the adversaries one by one.
 \begin{description}
  \item [{\normalfont Construction of $\sqrt{\adv}$:}]
  Suppose $\adv_{\CP}$ does not exist.
 $\sqrt{\adv}$ chooses an arbitrary $w \in \multgroup{\Delta}$. Now 
(s)he needs to choose a $u$ and to find a $\tau$ such that $\tau^\ell=w\in \multgroup{\Delta}$ where $\ell=\H(u)$.
So $\sqrt{\adv}$ fixes $u=i\| p_{k_1}\| \ldots \|p_{k_j}$ where $i$ is a leaf in $G_n$ and $\parent(i)=\{k_1, \ldots, k_j\}$.
Then $\sqrt{\adv}$ challenges $\adv$ on the leaf $i$ against the commitment $(x,w,N)$. 
 $\adv$ breaks the soundness with the probability $p_{win}$, 
 so $\adv$ must find a $\tau$ such that $\tau^\ell=w\in \multgroup{\Delta}$ where 
 $\ell=\H(i\| p_{k_1}\| \ldots \|p_{k_j})$. When $\adv$ returns $\tau$,   
 $\sqrt{\adv}$ outputs $\tau$ and wins the game $\rfg_\mathsf{H}$ with probability $p_{win}$.


\item [{\normalfont Construction of $\adv_{\CP}$:}]
Suppose $\sqrt{\adv}$ does not exist.
On a random challenge $\gamma_i$, $\adv_{\CP}$ breaks the soundness of the $\CP$-construction 
if $\adv_{\CP}$ succeeds to find the labels of the nodes 
that are required to construct the root label $\phi$ of $G_n$. 
By the second property of the $\CP$-graph (See Sect.~\ref{CP}), $\parent(\gamma_i)$ is necessary and sufficient 
to construct the root label. So, $\adv_{\CP}$ calls $\adv$ on the commitment $(x,\phi,N)$ against the $\gamma_i$.  
$\adv$ finds $\ell=\H(\gamma_i\| p_{k_1}\| \ldots \|p_{k_j})$ 
such that $\parent(\gamma_i)=\{k_1, \ldots, k_j\}$. 
These labels $p_{k_1}, \ldots ,p_{k_j}$ are consistent with probability   
$p_{win}$. So,  $\adv_{\CP}$ returns $p_{\gamma_i}$ and breaks the soundness of $\CP$ construction with the 
probability $p_{win}$.
\end{description}  
\qed
\end{proof}

The adaptive root assumption implies that $\sqrt{\adv}$ has only negligible
advantage. On the other hand, for $t$ number of challenges, the advantage of $\adv_\CP$ is upper bounded by $(1-\alpha)^t$
where $(1-\alpha)$ is the fraction of honest queries. We ignore the negligible slack in 
soundness caused by the collision in $\RO$. 
Therefore, it holds that,

\[
\Pr\left[
\begin{array}{l}
\phi\ne\mathsf{Solve}(\mathbf{pp},x)\\
\mathsf{Verify}(\mathbf{pp},\pi,x,\gamma,\phi)=1
\end{array}
\Biggm| \begin{array}{l}
\mathbf{pp}\leftarrow\mathsf{Gen}(1^\lambda,N)\\
\st\leftarrow\adv_0(1^\lambda,N,\mathbf{pp})\\
x\xleftarrow{\$}\mathcal{X}\\
(\phi,\pi)\leftarrow\adv_1(\st,x)\\
\end{array}
\right]\\
\le (1-\alpha)^t.
\]

Additionally, the factorization of $\Delta$ should not be known. 
\begin{lemma}\label{factorization}
 If $\adv$ knows the integer factorization of $\Delta=pq$ then 
  $\verify(\mathbf{pp},x,\phi',\gamma,\pi)=1$, 
 for any arbitrary $\rho'$ and $\phi'=g^{\rho'}\;\mathbf{mod}\; \Delta$.
 
\end{lemma}

\begin{proof}
Suppose $\adv$ chooses an arbitrary $\rho'$ and a commitment $\phi'$ without labelling $G_n$. 
Now, on an random challenge $\gamma_i$, $\adv$ 
finds $p_{\gamma_i}\leftarrow\H(i\| p_{k_1}\| \ldots \|p_{k_j})$ s.t $\parent(\gamma_i)=\sigma_i=\{k_1, \ldots, k_j\}$. 
Then (s)he computes 
$\tau_i'\leftarrow p_0^{ (\rho' \times p_{\gamma_i}^{-1})\;\mathbf{mod}\; \varphi(\Delta)}\;\mathbf{mod}\; \Delta$ 
where $\varphi(\Delta)=(p-1)(q-1)$ is the Euler's totient function and is the order of the group 
multiplicative group $(\mathbb{Z}/\Delta\mathbb{Z})^\times$. Therefore, 
$\phi'=(\tau_i')^{p_{\gamma_i}}\;\mathbf{mod}\; \Delta$. Therefore, we have
$\verify(\mathbf{pp},x,\phi,\gamma,\pi)=1$. 
 
\end{proof}

\subsection{Sequentiality} 

 By sequentiality we mean that no parallel adversary should be able to find a valid commitment and also proofs 
 in sequential time $<N-o(N)$  with a non-negligible probability.
\begin{theorem}\label{thm: sequentiality}
 Suppose $\adv$ be an adversary who breaks the sequentiality of this proposed $\PoSW$
 with probability $p_{win}$. Then there exists an attacker 
 $\adv'_{\CP}$ breaking the sequentiality of the $\CP$ construction,
 with the same probability $p_{win}$.
\end{theorem}

\begin{proof}
 $\adv'_{\CP}$ wants to label the $G_n$ consistently in time $\delta < N- o(N)$. 
 Here ``consistently" means given any random node $v \in G_n$,  $\adv'_{\CP}$ 
 must be able to figure out the label $p_v$ that are consistent with the labeling of $G_n$. 
 So, given the node $v$ on the input $x$, $\adv'_{\CP}$ calls
 the $\adv$ on the input $x$ and a random challenge $u$ such that $v \in \parent(u)$. 
 $\adv$ finds the label $p_v$ such that $p_u=\H(u \| \ldots \| p_v \| \ldots )$, 
 as it breaks the sequentiality of the proposed $\PoSW$. Thus $\adv'_{\CP}$ returns $p_v$ 
 on any random $v\in G_n$ and breaks the sequentiality of the $\CP$ construction
 with the same probability $p_{win}$.
 \qed
\end{proof}

So we have,
\[
\Pr\left[
\begin{array}{l}
\phi=\solve(\mathbf{pp},x)
\end{array}
\Biggm| \begin{array}{l}
\mathbf{pp}\leftarrow\mathsf{Gen}(1^\lambda,N)\\
\st\leftarrow\adv_0(1^\lambda,N,\mathbf{pp})\\
x\xleftarrow{\$}\mathcal{X}\\
\phi\leftarrow\adv_1(\st,x)
\end{array}
\right]
=\negl(\lambda).
\]

The most important corollary of the this theorem is,

\begin{description}

 \item [Source of Sequentiality:] \label{source}
 Theorem \ref{thm: sequentiality} never borrow the sequentiality of $f$. It is solely based on the 
 sequentiality of $\CP$-construction. So our $\PoSW$ stands sequential 
 even if $f$ is computable in $\O(1)$-time. In our design $f$ happens to be a sequential function but not necessarily.  
 We consider the random oracle $\mathsf{H}$ as the 
 only source of sequentiality and measure the sequentiality with the rounds of queries to $\mathsf{H}$.
 The map $f$ serves the purpose of fast verification using short proofs and the soundness. In fact, we want $f$ to be as efficient as possible 
 satisfying only a few properties mentioned in the Def.~\ref{Gen}.
%

 \end{description}

\section{Efficiency Analysis}\label{efficiency}
Here we discuss the efficiencies of both the prover $\mathcal{P}$ 
and the verifier $\mathcal{V}$ and the memory requirement for the commitment and the proof. 
We reiterate that we count the number of $\H$-queries only in order to estimate the 
sequential effort made by $\mathcal{P}$ and $\mathcal{V}$. 
The efforts to execute $f$ have been assumed to be efficient and need not be sequential.

First we need to determine the 
size of the group $(\mathbb{Z}/ \Delta \mathbb{Z})^\times$.

\subsection{Size of $(\mathbb{Z}/ \Delta \mathbb{Z})^\times$}
The number of bits to encode an element of the group $(\mathbb{Z}/ \Delta \mathbb{Z})^\times$ 
is $\log_2\Delta$. By Lemma~\ref{factorization}, it suffices to choose a $\Delta$ that takes $\Omega(N)$-time 
to be factored. It is shown in Sect.~\ref{subexp} that $N \in 2^{o(n)}$. 

So we describe the size of $\Delta=pq$ as a function of $\log N$ reported in~\cite{RSA}.
Here the sizes of the safe primes $p$ and $q$ are roughly $(\log_2\Delta)/2$. 
\begin{table}[h]
\caption{Size of $\Delta$ as a function of $\log N$.}
\label{tab : RSA}
 \centering
 \begin{tabular}{|l|r|}
    \hline
        $\;\log N\;$ & $\;\log \Delta\;$\\ \hline
        $\;80\;$  & $\;1024\;$       \\ \hline
        $\;112\;$ & $\;2048\;$       \\ \hline
        $\;128\;$ & $\;3072\;$       \\ \hline
	$\;192\;$ & $\;7680\;$       \\ \hline
        $\;256\;$ & $\;15360\;$       \\ \hline
  \end{tabular}
\end{table}

Further, assuming the generalized number field sieve method to be the most efficient heuristic 
for integer factorization, analytically, 
  \begin{align*}
          e^{((64/9)^{1/3}+o(1))(\ln \Delta)^{1/3}(\ln \ln \Delta)^{2/3}}  &\ge N \\
	  ((64/9)^{1/3}+o(1))(\ln \Delta)^{1/3}(\ln \ln \Delta)^{2/3}      &\ge \ln N  \\
	  \ln \Delta(\ln \ln \Delta)^2 &\ge \big(\frac{\ln N}{ (64/9)^{1/3}+o(1)}  \big)^3.
 \end{align*}

 Roughly $\tilde{\O}(\log \Delta)=\O(\log^3 N)$ where $\tilde{\O}(h(n))=h(n).\log^k h(n)$ for some $k \in \Z$.

\subsection{Proof Size} The size of the commitment $\phi$ is 
$\log_2 \Delta$. 
Each proof $\pi_i$ has two parts $s_{i}$ and $\tau_i$. The
$\sigma_i=\parent(\gamma_i)$ has at most $\log N$ number of $\lambda$-bit strings.
Thus $|\sigma_i|=\O(\lambda\log N )$. Additionally, $|\tau_i|=\log \Delta$. 
So $|\pi_i|=\O(\lambda\log N) + \log \Delta \approx \O(\lambda\log N+\log^3 N)$. 
Thus the dominating factor depends on the  $max\{\log N, \sqrt{\lambda}\}$.
So the total size of the proof $\pi=t \cdot \pi_i=\O(t(\lambda\log N+\log^3 N))$.
 
 \subsection{Prover's Efficiency}  
 
 A prover $\prv$ may adopt the similar approach like the $\CP$-construction. On random challenge $\prv$ needs to 
 find out consistent labels of its parent. $\prv$ may store $2^{m+1}$ labels at $m$ upmost levels into 
 the local information $\phi_\prv$ for some $0 \le m \le n=\log N$.
 \begin{itemize}
  \item for $m=0$,  $\prv$ labels $G_n$ spending $N$ rounds of queries once again.
  \item for $0 < m < n$, $\prv$ computes $2^{n-m+1}$ labels once again sequentially. 
  \item for $m=n$,  $\prv$ stores the entire labeled $\CP$-graph. No further queries are required.
 \end{itemize}

  In general, $\prv$ requires at most $\lambda N^{1-\beta}$ bits of memory for some $0<\beta<1$ if 
  (s)he is willing to spend only $tN^\beta$ queries in $\open$.
\subsection{Verifier's Efficiency} 

 This is where our $\PoSW$ outperforms all other existing $\PoSW$ demanding $\O(\log N)$ oracle queries.
 Note that in algorithm~\ref{alg : verify}, even a non-parallel $\vrf$ needs only a \emph{single} query to 
 the random oracle $\H$ in order to verify on a random challenge. 
 Thus our $\PoSW$ needs only $t$ queries in total instead of $\O(t \log N)$ queries as in the prior works.
 The memory requirement of $\vrf$ is same as the proof size.

\subsection{Our Contribution}\label{contributions}
 
 We summarize the comparison among all the $\PoSW$s in Table ~\ref{tab : PoSW}.

\begin{table*}[h]
\caption{Comparison among all the $\PoSW$s. Here $N$ is the rounds of queries to the random oracle $\H$.
Parallelism is upper bounded by $\O(\log N)$ processors. The effort for verification is the effort to verify a 
single challenge among $t$ such challenges in total. 
All the quantities may be subjected to $\O$-notation, if needed.}
\label{tab : PoSW}
 \centering
 \begin{tabular}{|l|r|r|r|r|r|r}
    \hline
    Schemes  & \textsf{Solve} & \textsf{Solve} &  \textsf{Verify} &  \textsf{Verify} & Proof-size \\
    (by authors) & (Sequential) & (Parallel) &  (Sequential) &  (Parallel) &  \\

    \hline
    
    Mahmoody et. al~\cite{Mahmoody2013Sequential}   & $N$  & $N$  &  $\log^2 N$  &  $\log N$  & $\lambda\log N$ \\
    \hline

    Cohen and Pietrzak~\cite{Cohen2018Simple}       & $N+ \sqrt{N}$ & $N $ &  $\log^2 N$  & $\log N$   &  $\lambda\log N$ \\
    \hline
    
    Abusalah et al.~\cite{Abusalah2019Reversible}   & $N$           & $N$  &  $\log^2 N$  &  $\log N$  &  $\lambda\log^2 N$ \\
    \hline 
    
    D\"{o}ttling et. al~\cite{Dottling2019Incremental}  & $N+ \sqrt{N}$ & $N$  &  $\log N$    &  $1$       & $\lambda\log N$\\
    \hline
    
    \textbf{Our work}       & $N+ \sqrt{N}$ & $N$  & \textbf{$1$} & \textbf{$1$} &  $max\{\lambda \log N,\log^3 N\}$ \\
    \hline
  \end{tabular}
\end{table*}

We show that it is possible to design a $\PoSW$ with non-parallel verifier that queries the random oracle $\H$ only \emph{once} 
per random challenge.
In particular, the verifier in our $\PoSW$ demands no parallelism but verifies using a single query to $\H$ in 
each of the $t$ trials. The state-of-the-art $\PoSW$~\cite{Dottling2019Incremental} verifies 
with a single \emph{round} of queries to $\H$ only if $\O(\log N)$ parallelism is available to the verifier. 
Without the parallelism they make $\O(\log N)$ rounds of queries to $\H$.
The prior works~\cite{Mahmoody2013Sequential,Cohen2018Simple,Abusalah2019Reversible} 
require $\O(\log N)$ rounds of queries $\H$ even in the presence of parallelism.

The fundamental observation is that the Merkle 
root commitment scheme, used in all the existing $\PoSW$s, mandates $\O(\log N)$ queries to verify.
Thus they require $\O(\log N)$ parallelism to reduce the verification time to a single round of queries.
So we replace the Merkle root commitment with a binary operation $\otimes$ and a binary map $f$
that immediately reduces the non-parallel verification effort to a single oracle query from $\O(\log N)$ queries.

Further, we mention an important class of $\PoSW$s that mandates
much more fine-grained notion of soundness.
These are called verifiable delay functions ($\vdf$)
that are nothing but $\PoSW$s but with unique proofs. That is, the soundness of such
$\PoSW$s has to be upper bounded by $\negl(\lambda)$ instead of $(1-\alpha)^t$ only.

There exist $\vdf$s based on
modulo exponentiation in the groups of unknown orders~\cite{Wesolowski2019Efficient,Pietrzak2019Simple} and 
isogenies over super-singular curves~\cite{Feo2019Isogenie}.
From the perspective design, the key difference between these $\vdf$s
and the $\PoSW$s based on random oracle is the source of sequentiality.
Although these $\vdf$s are more sound, their sequentiality are conditional.
For example, the $\vdf$s based on modulo exponentiation assume the condition that
$x^{2^N}\mod{\Delta}$ takes $\Omega(N)$-sequential time. To the best of our knowledge,
this condition has neither a proof nor a counter-example.
 On the other hand, all the
existing $\PoSW$s including ours achieve unconditional sequentiality in the random
oracle model.

\section{Generalizing The Function $f$}\label{design}
An obvious question to the above-mentioned $\PoSW$ is that is it necessary to be $f=g^a \mod \Delta$ always?
Or is there any other option to instantiate $f$? We 
give a concrete guideline to choose $f$ with respect to an operation $\otimes$ in general in a generic $\PoSW$.


The four algorithms that specify our $\PoSW$ are now described.

\subsection{The $\gen(1^\lambda,N)$ Algorithm}\label{Gen}
We need four sets $\mathbb{T},\mathbb{S}\subset\mathbb{R}\subset\mathbb{U}\subset\{0,1\}^*$
such that $\log_2 |\mathbb{S}| =  \poly(\lambda)$. 
Although all of these sets are exponentially large, 
we neither enumerate nor include them in the generated public parameters $\mathbf{pp}$. 
These sets are required in order to define 
the binary operation $\otimes$ and the binary map $f$.
The algorithm $\mathsf{Gen}$($1^\lambda$) outputs the public parameters
$\mathbf{pp}=\langle\textsf{H},G_n,\otimes,t, f\rangle$ with
the following meanings.
\begin{enumerate}
\item $\mathsf{H}:\{0,1\}^*\rightarrow\{0,1\}^\lambda$ is a random oracle (See Sect.~\ref{RO}).
\item $G_n$ is the $\CP$ graph (See Sect.~\ref{CP}) such that w.l.o.g., $N=2^{n+1}-1$.
\item $\otimes:\mathbb{U}\times\mathbb{U}\rightarrow\mathbb{U}$ is a commutative and associative binary operation 
      such that
 \begin{enumerate}
 \item $\langle \mathbb{U} , \otimes \rangle $ forms a group but 
 $\langle \mathbb{R} , \otimes \rangle $ forms a monoid. 
 \item It allows \emph{efficient} computation of,
 \begin{enumerate} \item the product $(s_i\otimes s_j)$ for all $s_i,s_j\in\mathbb{U}$.
                 \item the inverse $s_i^{-1}$ for all $s_i\in\mathbb{U}$.
                \end{enumerate}
 \item For any subset $\mathcal{S}_k=\{s_0, s_1, \ldots, s_k\}\subseteq\{0,1\}^{k\lambda}$ the product 
 $((s_0\otimes \ldots\otimes s_k)\otimes s_i^{-1})\in\mathbb{R}$ if and only if $s_i\in\mathcal{S}_k$.
 \end{enumerate}

\item $f:\mathbb{R}\times\mathbb{R}\rightarrow\mathbb{T}$ be an \emph{efficient} binary map such that,
\begin{enumerate}
\item For all  $g,a,b\in\mathbb{R}$, $f(f(g,a),b)=f(g,(a\otimes b))$.
 \item Given the result $f(g,a)$ and $a$, it is hard to find $g$. 
 \end{enumerate}
 \item $t$ as the number of random challenges.
\end{enumerate}

  None of the public parameters needs to be computed. 
  

\subsection{The $\solve(\mathbf{pp},x)$ Algorithm}
The prover,
\begin{enumerate}

 \item sample a random oracle $\H(\cdot)\stackrel{def}{=}\mathsf{H}(x \| \cdot)$ using $x$.
 \item compute $s_0= \H(0^n)$.
 \item initialize $\phi=s_0$ and $\rho=1_{\mathbb{U}}$.
 \item repeat for $1 \le i < N$ to label the graph $G_n$,
\begin{enumerate}
 \item $s_{i}=\H(i\| s_{k_1}\| \ldots \|s_{k_j})$ where $\parent(i)=\{k_1, \ldots, k_j\}$.
 \item $\rho =\rho \otimes s_i$.
 \item $\phi = f(\phi,s_i)$.
\end{enumerate}

\end{enumerate}

Announce the triple $(x,N,\phi)$ and stores the $\rho=(s_1\otimes \ldots \otimes s_N)$ locally. 
Here $\phi =f(\ldots f(f(s_0,s_1),s_2),\ldots, s_N)=f(s_0,\rho)$ serves as the commitment to the labels 
of $G_n$.

\subsection{The $\open(\mathbf{pp},x,N,\phi,\rho,\gamma)$ Algorithm}
In this phase $\mathcal{V}$ samples a set of challenge leaves 
$\gamma=\{\gamma_1, \gamma_2,\ldots,\gamma_{t}\}$ where each $\gamma_i\xleftarrow{\$}\Z_N$. 
For each $\gamma_i$, $\mathcal{P}$ needs to construct the proof 
$\pi=\{\sigma_i,\tau_i\}_{0 \le i \le t-1}$ by the following,

\begin{enumerate}
\item initialize an empty set $\pi=\{\emptyset\}$.
\item repeat for $1 \le i \le t$:
\begin{enumerate}
 \item assign $\sigma_i = \parent(\gamma_i)$.
 \item compute $\tau_i\leftarrow f(s_0, (\rho \otimes s_{\gamma_i}^{-1}))$.
 \item assign $\pi_i =\{\sigma_i,\tau_i\}$.
\end{enumerate}
\end{enumerate}

$\mathcal{P}$ can not compute $\tau_i\leftarrow f(\phi, s_{\gamma_i}^{-1})$ because 
the inverse $s_{\gamma_i}^{-1}\notin\mathbb{R}$ as $\langle \mathbb{R}, \otimes \rangle$ is a monoid. 
%
%
%
%
%

\subsection{The $\verify(\mathbf{pp},x,N,\gamma,\pi,\phi)$ Algorithm}

The $\vrf$ runs these steps.
\begin{enumerate}
\item samples a random oracle $\H(\cdot)\stackrel{def}{=}\mathsf{H}(x \| \cdot)$ using $x$.
 \item computes $s_0= \H(0^n)$.
\item repeat for $1 \le i < t$:
\begin{enumerate}
 \item compute $s_{\gamma_i}=\H(i\| s_{k_1}\| \ldots \|s_{k_j})$ where $\parent(\gamma_i)=\sigma_i=\{k_1, \ldots, k_j\}$.
 \item assign a flag $r_i$ if and only if $\phi\stackrel{?}{=}f(\tau_i,s_{\gamma_i})$.
\end{enumerate}
 \item Accept if $(r_0 \land r_1 \land \ldots \land r_{t-1})\stackrel{?}{=}\top$, rejects otherwise.
  \end{enumerate}

 The security and the efficiency analysis are exactly same as in the Sect.~\ref{security} and Sect.~\ref{efficiency}.
 It just needs to replace $\times$ with $\otimes$ and taking $f$ in general. 
 Still we reiterate exactly the same proofs for the sake of completeness.

\begin{theorem}
 The generalized {\normalfont $\PoSW$} is correct.
\end{theorem}
\begin{proof}
 In the algorithm $\solve$, since $\mathsf{H}$ always outputs a fixed element, 
 $s_0$ is uniquely determined by the input $x$. Moreover in $\verify$, 
 for each $1 \le i < t$,
 \begin{align*}
 f(\tau_i, s_{\gamma_i}^{-1}) 
 &= f(f(s_0, (\rho\otimes s_{\gamma_i}^{-1})),s_{\gamma_i})\\ 
 &=f(s_0,(\rho\otimes s_{\gamma_i}^{-1}\otimes s_{\gamma_i}))\\
                                           &= f(s_0,\rho)\\
                                           &=\phi. 
 \end{align*}
 So, the correct enumeration of the sequence $\sigma$ and evaluation $\phi$ assigns $f_i=\top$.
 Then for each $i, f_i$ must be $\top$ which results into $f=\top$ in 
 $\verify$. So $\mathcal{V}$ has to accept it. 
 \qed
\end{proof}

\begin{theorem}\label{thm: soundness}
 Suppose $\adv$ be an adversary who breaks the soundness of this generalized $\PoSW$
 with probability $p_{win}$. Then there exists \emph{any} one of these two attackers,
 \begin{enumerate}
  \item $\adv^{-1}$ inverting $f$ on its first argument $g$, 
  \item $\adv_{\CP}$ breaking the soundness of the $\CP$ construction,
 \end{enumerate}
 with the same probability $p_{win}$.
\end{theorem}

\begin{proof}
We give the adversaries one by one.
 \begin{description}
  \item [{\normalfont Construction of $\adv^{-1}$:}]
  Suppose $\adv_{\CP}$ does not exist.
 $\adv^{-1}$ chooses an arbitrary $\phi \in \mathbb{T}$ but 
 $u=i\| s_{k_1}\| \ldots \|s_{k_j}$ where $i$ is a leaf in $G_n$ and $\parent(i)=\{k_1, \ldots, k_j\}$.
 Then $\adv^{-1}$ challenges $\adv$ on the leaf $i$ against the commitment $(x,w,N)$. 
 $\adv$ breaks the soundness with the probability $p_{win}$. 
 So $\adv$ computes $s_i=\H(i\| s_{k_1}\| \ldots \|s_{k_j})$ and 
 finds a $\tau$ such that $f(\tau,s_i)=\phi$.  $\adv^{-1}$ returns $\tau$ to invert
 $f$ on its first argument with probability $p_{win}$.

\item [{\normalfont Construction of $\adv_{\CP}$:}]
Suppose $\adv^{-1}$ does not exist.
On a random challenge $\gamma_i$, $\adv_{\CP}$ breaks the soundness of the $\CP$-construction 
if $\adv_{\CP}$ succeeds to find the labels of the nodes 
that are required to construct the root label $\phi$ of $G_n$. 
By the second property of the $\CP$-graph (See Sect.~\ref{CP}), $\parent(\gamma_i)$ is necessary and sufficient 
to construct the root label. So, $\adv_{\CP}$ calls $\adv$ on the commitment $(x,\phi,N)$ against the $\gamma_i$.  
$\adv$ finds $s_{\gamma_i}=\H(\gamma_i\| s_{k_1}\| \ldots \|s_{k_j})$ 
such that $\parent(\gamma_i)=\{k_1, \ldots, k_j\}$. 
These labels $s_{k_1}, \ldots ,s_{k_j}$ are consistent with probability   
$p_{win}$. So,  $\adv_{\CP}$ returns $s_{\gamma_i}$ and breaks the soundness of $\CP$ construction with the 
probability $p_{win}$.
\end{description} 
\qed
\end{proof}

\begin{theorem}
 Suppose $\adv$ be an adversary who breaks the sequentiality of this generalized $\PoSW$
 with probability $p_{win}$. Then there exists an attacker 
 $\adv'_{\CP}$ breaking the sequentiality of the $\CP$ construction,
 with the same probability $p_{win}$.
\end{theorem}

\begin{proof}
 $\adv'_{\CP}$ wants to label the $G_n$ consistently in time $\delta(N) < N$. 
 Here ``consistently" means given any random node $v \in G_n$,  $\adv'_{\CP}$ 
 must be able to figure out the label $s_v$ that are consistent with the labeling of $G_n$. 
 So, given the node $v$ on the input $x$, $\adv'_{\CP}$ calls
 the $\adv$ on the input $x$ and a random challenge $u$ such that $v \in \parent(u)$. 
 $\adv$ finds the label $s_v$ such that $s_u=\H(u \| \ldots \| s_v \| \ldots )$, 
 as it breaks the sequentiality of the proposed $\PoSW$. Thus $\adv'_{\CP}$ returns $s_v$ 
 on any random $v\in G_n$ and breaks the sequentiality of the $\CP$ construction
 with the same probability $p_{win}$.
 \qed
\end{proof}

\section{Conclusion and Open Problem}\label{conclusion}
This paper presents a proof 
of sequential work that queries the random oracle only once while verifying. 
Our $\PoSW$ is based on the one in~\cite{Cohen2018Simple}, however uses two 
additional primitives i.e., the operation $\otimes$ and the map $f$ in its design.
We have been able to show that even a non-parallel verifier needs only a single oracle query to verify. 
So the effort for verification reduces 
to a single query from logarithmically proportional queries to $N$ (time) as in the existing $\PoSW$s. 
The key idea is to replace the Merkle root based commitment with the operation $\otimes$ and the map $f$. 

Our $\PoSW$ is proven to be correct, sound and sequential.
Finally we give a proper guideline to choose $\otimes$ and $f$.
However, it turns out to be a nice open question that how far we can minimize the computation time of the map $f$ 
maintaining the guidelines. In particular only the soundness depend on the choice of $f$. 
So we would always like to have $f$ that is as fast as possible. Is there any lower bound of this computation time beyond 
which the safeguard will be violated or we are free to have any $f$. If yes, then what are the other options?

\bibliographystyle{splncs04}
\bibliography{ref}

\end{document}